\documentclass[runningheads]{llncs}
\usepackage{multirow}
\usepackage{graphicx}
\usepackage{subfig}
\usepackage{amsmath,amsbsy,amssymb,latexsym}
\usepackage{caption}
\usepackage{color,url}
\usepackage{xspace} 
\usepackage{algorithm}
\usepackage{algorithmic}
\usepackage[algo2e]{algorithm2e}
\usepackage{float}
\usepackage{hyperref}

\newcommand{\myparatight}[1]{\noindent{\bf {#1}:}~}
\newenvironment{packeditemize}{\begin{list}{$\bullet$}{\setlength{\itemsep}{1pt}\addtolength{\labelwidth}{20pt}\setlength{\leftmargin}{\labelwidth}\setlength{\listparindent}{\parindent}\setlength{\parsep}{0pt}\setlength{\topsep}{0pt}}}{\end{list}}
\newcommand{\Alan}{}
\newcommand{\alan}{}

\usepackage{hyperref}

\makeatletter
\newcommand{\printfnsymbol}[1]{%
  \textsuperscript{\@fnsymbol{#1}}%
}
\makeatother

\begin{document}
%

\title{SybilBlind: Detecting Fake Users in Online Social Networks without Manual Labels}
\author{Binghui Wang \and
Le Zhang \and
Neil Zhenqiang Gong\thanks{The first two authors contributed equally to this work.}}
\institute{ECE Department, Iowa State University \\
\email{\{binghuiw, lezhang, neilgong\}@iastate.edu}
}



\maketitle              

\begin{abstract}
Detecting fake users (also called Sybils) in online social networks is a basic 
security research problem. State-of-the-art approaches rely on a large amount 
of manually labeled users as a training set. These approaches suffer from three key limitations: 1) it is time-consuming and costly to manually label a large training set, 2) 
they cannot detect new Sybils in a timely fashion, and 3) they are vulnerable to Sybil attacks that leverage information of the training set. 
 In this work, we propose \emph{SybilBlind}, a structure-based Sybil detection framework that does not rely on a manually labeled training set.  SybilBlind works under the same threat model as state-of-the-art structure-based methods.  
 We demonstrate the effectiveness of SybilBlind using 1) a social network with synthetic Sybils and 2) two Twitter datasets with 
real Sybils. 
\alan{For instance, SybilBlind achieves an AUC of 0.98 on a Twitter dataset.}


\end{abstract}
\section{Introduction}


Online social networks (OSNs) are known to be vulnerable to \emph{Sybil attacks}, in which attackers maintain a large number of fake users (also called Sybils).
For instance, 10\% of Twitter users were fake~\cite{Twittersybil}. 
Attackers can leverage Sybils to perform various malicious activities such as manipulating presidential election~\cite{election},  influencing stock market~\cite{stock}, distributing spams and phishing URLs~\cite{Thomas11}, etc.. 
Therefore, Sybil detection in OSNs is an important research problem. 



Indeed, Sybil detection has attracted increasing attention from multiple research communities such as security, networking, and data mining. 
Among various approaches, structure-based ones~\cite{Yu06,Yu08,Danezis09,Viswanath10,Yang11-sybil,sybilrank,sybildefender,Yang12-spam,integro,sybilbelief,sybilfuse,sybilscar,sybilwalk,GANG} have demonstrated promising results.  For instance, SybilRank~\cite{sybilrank} and Integro~\cite{integro} were deployed to detect a large amount of Sybils in Tuenti, the largest OSN in Spain. 
\alan{SybilSCAR~\cite{sybilscar} was shown to be effective and efficient in detecting Sybils in Twitter.} 
State-of-the-art structure-based approaches adopt the following machine learning paradigm:  they first require an OSN provider to collect a large manually labeled training set consisting of  labeled benign users and/or  labeled Sybils; then they   learn a model to  distinguish between benign users and Sybils; finally, the model is used to detect Sybils.

Such paradigm of relying on a manually labeled training set suffers from three key limitations. First, it is time-consuming and costly to obtain a large manually labeled training set. 
We note that OSN providers could outsource 
manual labeling to crowdsourcing services like Amazon Mechanical Turk~\cite{Wang13}.
However,  crowdsourcing manual labeling requires disclosing user information to ``turkers", which raises privacy concerns. 
Moreover, attackers could act as ``turkers" to adversarially mislabel users.  
OSNs often allow users to flag other users as Sybils. However, similar to crowdsourcing, Sybils could adversarially mislabel benign users as Sybils.
Second, attackers can launch new Sybil attacks when the old ones were taken down.
 It takes time for human workers to manually label a training set for the new attacks.
 As a result,  some benign users might already be attacked before the new attacks were detected. 
 Third, using a manually labeled training set makes these approaches vulnerable to 
 Sybil attacks that leverage the information of the training set~\cite{TemporalDynamicsCCS15Liu}. The key intuition is that once an attacker knows or infers the training set, he can perform better attacks over time. 
 \alan{Our method  is secure against such attacks as it does not rely on labeled users. }

\myparatight{Our work} In this work, we propose SybilBlind, a structure-based framework,  to detect Sybils without relying on a manually labeled training set, under the same threat model as state-of-the-art structure-based methods (See Section~\ref{threatmodel}).
%
Our key idea is to sample some users from an OSN, randomly assign labels (i.e., \emph{benign} or \emph{Sybil}) to them, and treat them as if they were a training set without actually manually labeling them.  
Such randomly sampled training set could have various levels of label noise, 
where a user's randomly assigned label is noisy if it is different from the user's true label. 
Then, we take the noisy training set as an input to a state-of-the-art Sybil detection method (e.g., SybilSCAR~\cite{sybilscar} in our experiments) that is relatively robust to label noise (i.e., performance does not degrade much with a relatively low fraction of noisy labels) to detect Sybils. 
We define a \emph{sampling trial} as the process that we randomly sample a noisy training set 
and use a state-of-the-art Sybil detection method to detect Sybils via taking the sampled training set as
an input. Since state-of-the-art Sybil detection methods can only accurately detect Sybils in the sampling trials where the sampled training sets have relatively low label noise, we repeat for multiple sampling trials and we design an aggregator to aggregate the results in the multiple sampling trials.  

A key challenge of our SybilBlind framework is how to aggregate the results in multiple sampling trials. 
For instance, one natural aggregator is to average the results in multiple sampling trials. 
Specifically, in each sampling trial, we have a probability of being a Sybil for each user.  We average the probabilities over multiple sampling trials for each user and use the averaged probability to 
classify a user to be benign or Sybil. However, we demonstrate, both theoretically and empirically, that such average aggregator achieves an accuracy that is close to random guessing.   
To address the aggregation challenge, we design a novel aggregator. Specifically, we design two new metrics called \emph{homophily} and \emph{one-side entropy}. 
In a sampling trial where Sybils are accurately detected, both homophily and one-side entropy are large. With the two metrics, our aggregator identifies the sampling trials in which the sampled training sets have low label noise and Sybils are accurately detected. Then, we compute an \emph{aggregated probability} of being 
a Sybil for every user from these sampling trials and use the aggregated probabilities to detect Sybils.

We evaluate SybilBlind both theoretically and empirically. Theoretically, we analyze the required number of sampling trials. Empirically, we perform evaluations using 1) a social network with synthesized Sybils, 2) a small  Twitter dataset (8K users and 68K edges) with real Sybils, and 3) a large Twitter dataset (42M users and 1.2B edges) with real Sybils. Our results demonstrate that SybilBlind is accurate, e.g., on the small Twitter dataset, SybilBlind achieves an AUC of 0.98.  Moreover, we adapt a community detection method and state-of-the-art Sybil detection method SybilSCAR~\cite{sybilscar} to detect Sybils when a manually labeled training set is unavailable. Our empirical evaluations demonstrate that SybilBlind substantially outperforms these adapted methods.

Our key contributions are summarized as follows:
\begin{packeditemize}
\item We propose SybilBlind, a structure-based framework, to detect Sybils in OSNs without relying on a manually labeled training set. 

\item We design a novel aggregator based on homophily and one-side entropy to aggregate results in multiple sampling trials. 

\item We evaluate SybilBlind both theoretically and empirically, as well as compare it with Sybil detection methods that we adapt to detect Sybils when no manually labeled training sets are available.  
Our empirical results demonstrate the superiority of SybilBlind over the adapted methods. 
  
\end{packeditemize}

\section{Related Work}
\label{relatedwork}

\subsection{Structure-based Approaches} 
One category of Sybil detection approaches leverage the global structure of the social network~\cite{Yu06,Yu08,Danezis09,Viswanath10,Yang11-sybil,sybilrank,sybildefender,Yang12-spam,integro,sybilbelief,robustspammer,sybilfuse,sybilscar,sybilwalk,GANG,sybilscarJournal}.
These approaches require a manually labeled training dataset, from which they propagate label information among the social network via leveraging the social structure.

\myparatight{Using random walks or Loopy Belief Propagation (LBP)} Many structure-based approaches~\cite{Yu06,Yu08,Danezis09,sybilrank,Yang12-spam,integro,sybilwalk} leverage random walks to propagate label information. 
SybilGuard~\cite{Yu06}, SybilLimit~\cite{Yu08}, and SybilInfer~\cite{Danezis09} only require one labeled benign user. However, they achieve limited performance and are not scalable to large-scale OSNs.  SybilRank~\cite{sybilrank} and \'{I}ntegro~\cite{integro} are state-of-the-art random walk based approaches, and they were successfully applied to detect a large amount of Sybils in Tuenti, the largest OSN in Spain. 
However, they require a large number of manually labeled benign users; and \'{I}ntegro even further requires a large number of labeled victims and non-victims,  which were used to learn a binary victim-prediction classifier. A user is said to be a victim if the user is connected with at least a Sybil. 
SybilBelief~\cite{sybilbelief}, Fu et al.~\cite{robustspammer}, GANG~\cite{GANG}, and SybilFuse~\cite{sybilfuse} leverage probabilistic graphical model techniques. Specifically, they model a social network as a pairwise Markov Random Fields. Given a training dataset, they leverage LBP
to infer the label of each remaining user. 

Recently, Wang et al.~\cite{sybilscar,sybilscarJournal} proposed a local rule based framework to unify random walk and LBP based approaches. Under this framework, a structure-based Sybil detection method essentially iteratively applies a certain local rule to each user to propagate label information. Different Sybil detection methods use different local rules. Moreover, they also proposed a new local rule, based on which they designed SybilSCAR that achieves state-of-the-art performance both theoretically and empirically. For instance, SybilSCAR achieves the tightest asymptotic bound on the number of Sybils per attack edge that can be injected into a social network without being detected~\cite{sybilscarJournal}.  \alan{However, as we demonstrate in our experiments on Twitter, SybilSCAR requires a large training dataset in order to achieve an accurate Sybil detection performance.}




\myparatight{Using community detection algorithms} Viswanath et al.~\cite{Viswanath10} showed that  Sybil detection  
can be cast as a community detection problem. The authors found that detecting local community around a labeled benign user had equivalent results to approaches such as SybilLimit and SybilInfer. 
Cao et al.~\cite{sybilrank} showed that SybilRank significantly outperforms community detection approaches. 
Moreover, Alvisi et al.~\cite{alvisiSybil13} demonstrated a vulnerability of the local community detection algorithm 
adopted by Viswanath et al.~\cite{Viswanath10}  by carefully designing an attack.

\myparatight{Summary} State-of-the-art structure-based approaches (e.g., SybilRank, SybilBelief, and SybilSCAR) require a large manually labeled training dataset. These approaches suffer from three key limitations as we discussed in Introduction. 


\subsection{Other Approaches}
Approaches in this direction~\cite{Wang10,spam:acsac10,benevenuto2010detecting,Yang_RAID11_TwitterML,kontaxis2011detecting,Yang11-sybil,Thomas11,gao2012towards,Wang13Clickstream,Song11} leverage various user-generated contents (e.g., tweets), behaviors (e.g., the frequency of sending tweets), and local social structures (e.g., how a user's friends are connected).  Most studies in this direction~\cite{Wang10,spam:acsac10,benevenuto2010detecting,Thomas11,gao2012towards,Song11} treat Sybil detection as a supervised learning problem; they extract various features from user-generated contents, behaviors, and local social structures, and they learn machine learning classifiers using a training dataset; the learnt classifiers are then used to classify each remaining user to be benign or Sybil.  For instance, Yang et al.~\cite{Yang11-sybil} proposed   local social structure based features  such as  the frequency that a user sends friend requests to others, the fraction of outgoing friend requests that are accepted, and the clustering coefficient of a user. 
One limitation  of these approaches is that Sybils can manipulate users' profiles to evade detection. For instance, a Sybil can link to many Sybils to manipulate its local social structure as desired. 
However, although these approaches are easy to evade, we believe that they can be used as a first layer to filter some basic Sybils and increase attackers' costs of performing Sybil attacks. Moreover, these approaches are complementary to approaches that leverage global social structures, and they can be used together in practice. 
\Alan{For instance, we can treat the outputs of these approaches as users' prior probabilities. Then, we can leverage structure-based methods, e.g., SybilSCAR~\cite{sybilscar}, to detect Sybils by iteratively propagating the priors among a social network.}





\section{Problem Definition}

\subsection{Structure-based Sybil Detection without Manual Labels}

Suppose we are given an undirected social network $G=(V,E)$,\footnote{Our framework can also be generalized to directed social networks.} 
where a node in $V$ corresponds to a user in an OSN 
and an edge $(u,v)$ represents a certain relationship between $u$ and $v$. 
For instance, on Facebook, an edge between $u$ and  $v$ could 
mean that $u$ is in $v$'s friend list and vice versa.
On Twitter, an edge $(u,v)$ could mean that $u$ and $v$ follow each other. 
We consider Sybil detection without a manually labeled training dataset, which we call \emph{blind Sybil detection}.  


\begin{definition}[Blind Sybil Detection] 
Suppose we are given a social network.  Blind  Sybil detection is to  classify each node to be benign or Sybil without a manually labeled training dataset.   
\end{definition}

\subsection{Threat Model}
\label{threatmodel}
We call the subnetwork containing all benign nodes and edges between them the \emph{benign region}, 
and we call the subnetwork containing all Sybil nodes and edges between them the \emph{Sybil region}. 
The edges between the two regions are called \emph{attack edges}. 
We consider the following threat model, which is widely adopted by existing structure-based methods.

\myparatight{Connected-Sybil attacks} We consider that Sybils 
are connected 
among themselves. 
In order to leverage Sybils to launch various malicious activities, an attacker often needs to first link his/her created Sybils to benign users. One attack strategy is that each Sybil aggressively sends friend requests to a large number of users (or follow a large number of users) that are randomly picked~\cite{Yang11-sybil}. In these attacks, although
some benign users (e.g., social capitalists~\cite{Ghosh12}) will accept such friend requests with a relatively high probability, making the Sybils embed to the benign region, most benign users will not accept these friend requests~\cite{Ghosh12}. As a result,  Sybils that are created using this attack strategy often have low ratios of accepted friend requests (or ratios of being followed back), as well as low clustering coefficients because most users that link to a Sybil might not be connected with each other. Therefore,  such Sybils can be detected by machine learning classifiers that use these structural features, as was shown by Yang et al.~\cite{Yang11-sybil} on RenRen, a large OSN in China.

In this paper, we consider that Sybils created by an attacker are connected (i.e., \emph{connected-Sybil attack}), so as 
to manipulate their structural features to evade the detection of structural feature based classifiers. 
Such connected-Sybil attacks were formally discussed by Alvisi et al.~\cite{alvisiSybil13}, are required by previous structure-based methods~\cite{Yu06,Yu08,Danezis09,Viswanath10,sybilrank,Yang12-spam,integro,sybilbelief,sybilscar,sybildefender}. 
Note that Sybils in Tuenti~\cite{sybilrank}, the largest OSN in Spain, are densely connected. 
Moreover, the datasets we used in our experiments also show that most of the Sybils are connected. For instance, in our large Twitter dataset, 85.3\% Sybils are connected to form a largest connected component with an average degree 24. 

\myparatight{Limited number of attack edges} Intuitively, most benign users would not establish \emph{trust} relationships with Sybils. We assume that the number of attack edges is relatively smaller, compared to the number of edges in the benign region and the Sybil region. This assumption is required by all previous structure-based methods~\cite{Yu06,Yu08,Danezis09,Viswanath10,sybilrank,Yang12-spam,integro,sybilbelief,sybilscar,sybildefender} except  \'{I}ntegro~\cite{integro}. 
\'{I}ntegro assumes  the number of victims (a victim is a node having attack edges) is small and  victims can be accurately detected.
The number of attack edges in Tuenti was shown to be relatively small~\cite{sybilrank}. Service providers can limit the number of attack edges via approximating trust relationships between users, e.g., looking into user interactions~\cite{wilson:eurosys09}, inferring tie strengths~\cite{gilbert:chi09}, and asking users to rate their social friends~\cite{sybildefender}.  
We note that in the large Twitter dataset we used in our experiments, only 1.5\% of the total edges are attack edges.   

For connected-Sybil attacks, limited number of attack edges is equivalent to the \emph{homophily} assumption, i.e., if we randomly sample an edge $(u,v)$ from the social network,
then $u$ and $v$ have the same label with high probability. In the following, we use homophily and limited number of attack edges interchangeably.


\myparatight{Benign users are more than Sybils} We assume that  Sybils are less than  benign users in the OSN. 
An attacker often leverages only tens of thousands of compromised hosts to create and manage Sybils~\cite{SybilUnderground13}. If an attacker registers and maintains a large number of Sybils on each compromised host, the  OSN provider can easily detect these Sybils via IP-based methods. 
In other words, to evade detection by IP-based methods, each compromised host can only maintain a limited number of Sybils. Indeed, Thomas et al.~\cite{SybilUnderground13} found that a half of compromised hosts under an attacker's control maintain less than 10 Sybils. As a result, 
in OSNs with
tens or hundreds of millions of benign users, the number of Sybils is smaller than that of benign users.
For instance, it was reported that  10\% of Twitter users were Sybils~\cite{Twittersybil}. 
Our method leverages this assumption to break the symmetry between the benign region and the Sybil region.

\section{Design of SybilBlind}
\subsection{Overview}
%

Figure~\ref{overview} overviews SybilBlind.  
SybilBlind consists of three components, i.e., \emph{sampler}, \emph{detector}, and \emph{homophily-entropy aggregator (HEA)}.  Sampler  samples two subsets of nodes  from the social network, and constructs a training set by assigning a label of benign to nodes in one subset and a label of Sybil to nodes in the other subset. 
 The detector takes the sampled noisy training set as an input and produces a probability of being Sybil for each node. The detector can be any structure-based Sybil detection method (e.g., SybilSCAR~\cite{sybilscar} in our experiments) that is relatively robust to label noise in the training set. SybilBlind repeats this sampling process for multiple trials, and it leverages  a homophily-entropy aggregator to identify the sampling trials in which the detector accurately detects Sybils. Finally, SybilBlind computes an \emph{aggregated probability} of being 
 Sybil for every node using the identified sampling trials.  
 
\begin{figure}[!tbh]
  \centering
  \begin{minipage}[b]{0.48\textwidth}
    \includegraphics[width=\textwidth]{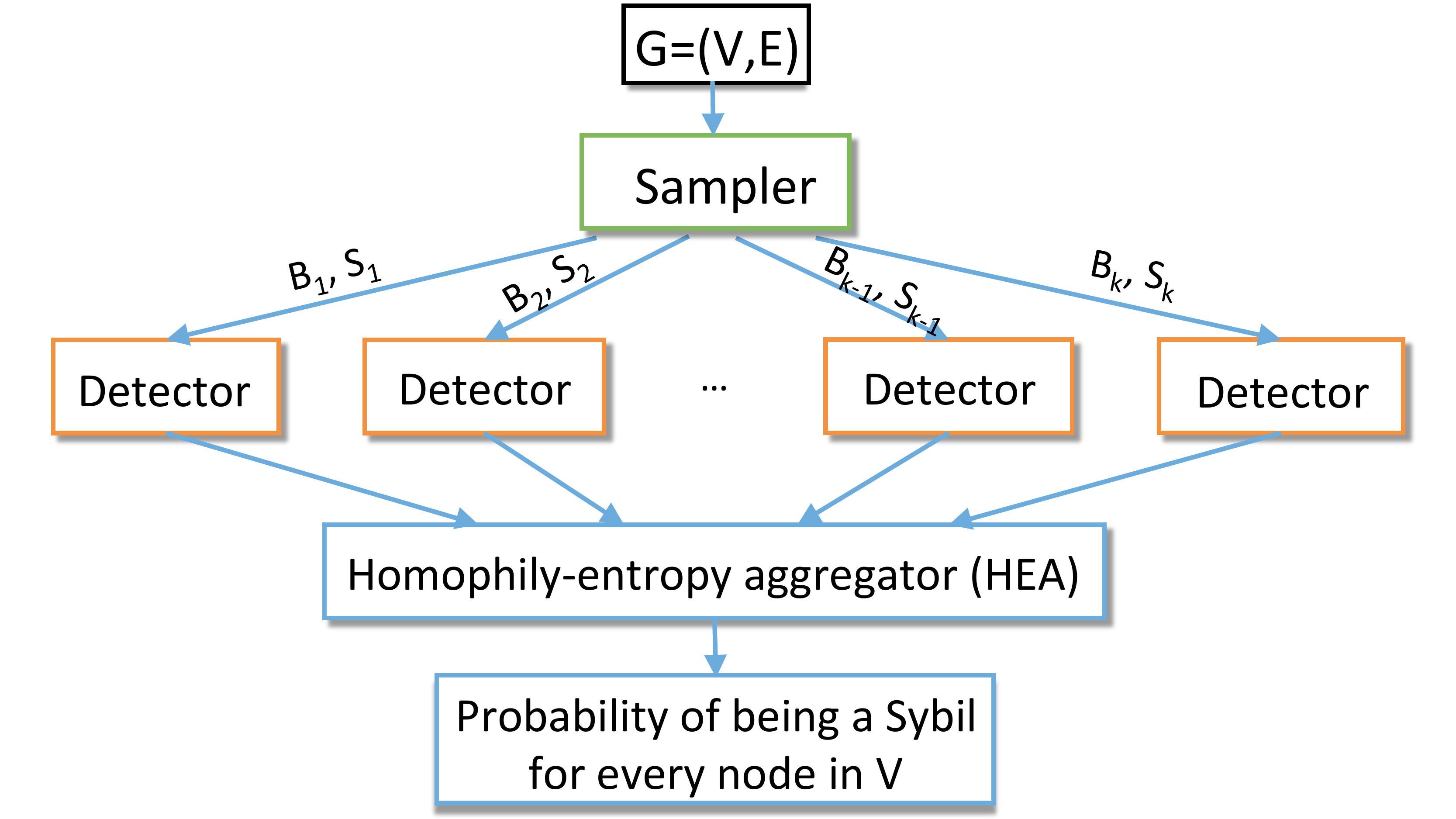}
    \caption{Overview of SybilBlind.}
    \label{overview}
  \end{minipage}
  \hfill
  \begin{minipage}[b]{0.48\textwidth}
    \includegraphics[width=\textwidth]{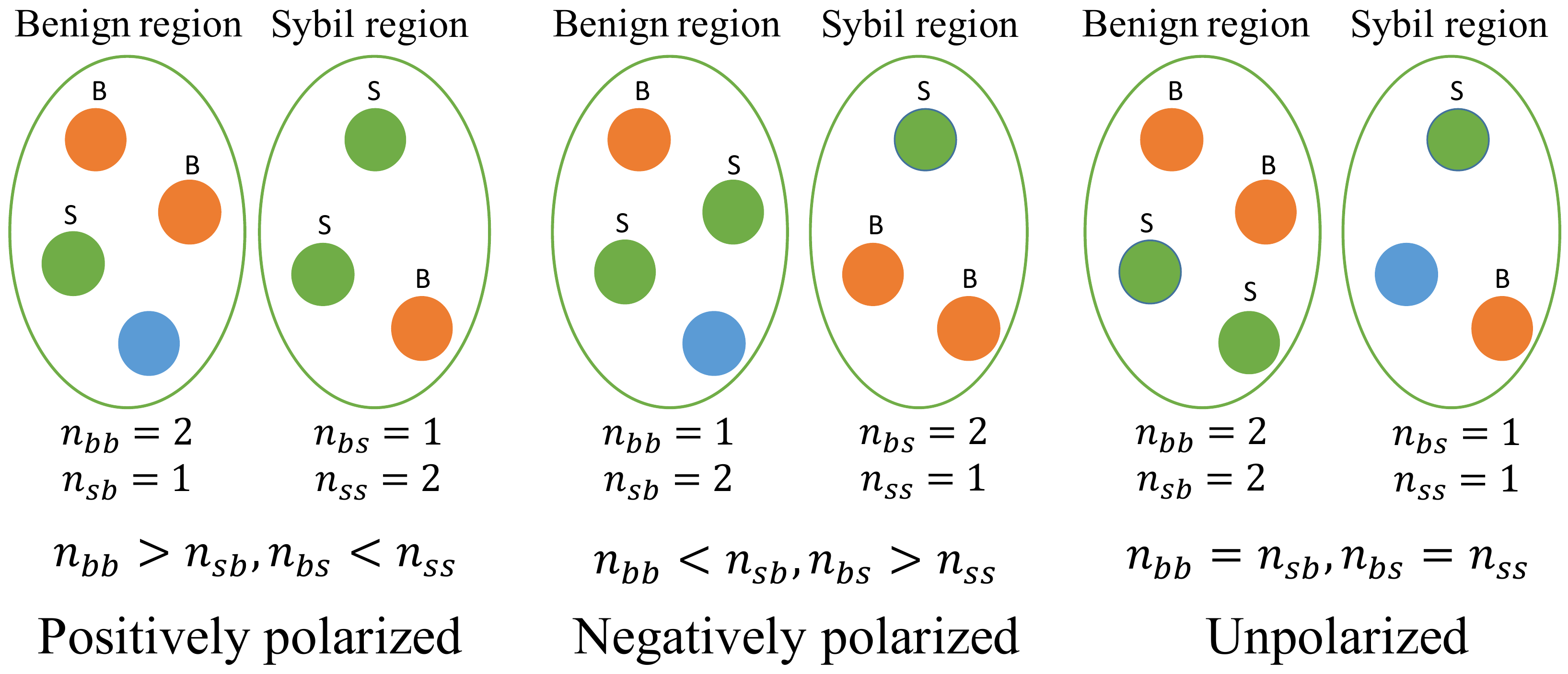}
    \caption{Three scenarios of our sampled nodes with a sampling size 3.}
	\label{sampler}
  \end{minipage}
\end{figure}

\subsection{Sampler}
In each sampling trial, our sampler samples two subsets of nodes from the set of nodes $V$, which are denoted as $B$ and $S$, respectively.  
Moreover, for simplicity, we consider the two subsets have the same number of nodes, i.e., $n=|B|=|S|$, and we call $n$ the \emph{sampling size}. 
\Alan{We note that it would be a valuable future work to apply our SybilBlind framework to subsets $B$ and $S$ with different sizes.} 

The subset $B$ (or $S$) might consist of both benign nodes and Sybils. 
For convenience, we denote by $n_{bb}$ and $n_{bs}$ respectively the number of benign nodes and the number of Sybils in $B$; and we denote by $n_{sb}$ and $n_{ss}$ respectively the number of benign nodes and the number of Sybils in $S$. We categorize the sampled nodes into three scenarios because they have different impacts on the performance of the detector. Figure~\ref{sampler} shows one example of the three scenarios, where $n=3$. 
 The three scenarios are as follows:


\begin{packeditemize}
\item {\bf Positively polarized:} In  this scenario,  
the number of benign nodes in $B$ is larger than the number of benign nodes in $S$, while the number of Sybils in $B$ is smaller than the number of Sybils in $S$. Formally, we have $n_{bb} > n_{sb}$ and $n_{bs} < n_{ss}$.  
\item {\bf Negatively polarized:} In  this scenario, 
$B$ includes a smaller  number of benign nodes than $S$, while $B$ includes a larger number of Sybils than $S$. Formally, we have $n_{bb} < n_{sb}$ and $n_{bs} > n_{ss}$.  
\item {\bf Unpolarized:} In  this scenario, 
the number of benign (or Sybil) nodes in $B$ equals the number of benign (or Sybil) nodes in $S$. Formally, we have $n_{bb} = n_{sb}$ and $n_{bs} = n_{ss}$.  
\end{packeditemize} 

\Alan{Note that since the two subsets $B$ and $S$ have the same number of nodes, we only have the above three scenarios.} We construct a training set using the sampled $B$ and $S$. Specifically, we assign a label of benign to nodes in $B$ and a label of Sybil to  nodes in $S$. 
Such training set could have label noise. In particular, in a sampling trial that is positively polarized, a majority of sampled nodes are assigned labels that match their true labels;
while in a sampling trial that is negatively polarized,
a majority of sampled nodes are assigned labels that do not match their true labels.

\subsection{Detector}
The detector takes a (noisy) training set as an input and produces a probability of being Sybil for every node (including the sampled nodes in the training set). The requirement for the detector is to be relatively robust to label noise in the training set. 
In this work, we adopt SybilSCAR~\cite{sybilscar} as the detector as it was shown to achieve state-of-the-art accuracy and robustness to label noise. However, we stress that 
our framework is extensible to use other structure-based Sybil detection methods as the detector. In particular, if a better structure-based Sybil detection method that uses a manually labeled training set is designed in the future, we can use it as the detector to further improve SybilBlind.

Next, we briefly review SybilSCAR. Given the sampled training set, SybilSCAR assigns a prior probability $q_u$ of being Sybil for every node $u$. 
Specifically, 
{\small
\begin{align}
{q}_u = 
\begin{cases} 
0.5 + \theta &\text{if } u \in S \\ \nonumber
0.5 - \theta &\text{if } u \in B \\ \nonumber
0.5 &\text{otherwise,} \nonumber
\end{cases} 
\end{align}
}%
where $0<\theta<0.5$ is a parameter to consider label noise.

\alan{Given the priors, SybilSCAR iteratively computes the probability $p_u$ of being Sybil for every node $u$ until convergence. Specifically, initially we have $p_u^{(0)} = q_u$. In the $t$th iteration, for each node $u$, we have:
{\small
\begin{align}
{p}_u^{(t)}  = {q}_{u} + 2 ({w} - 0.5) \sum_{v \in \Gamma(u)} ({p}_{v}^{(t-1)} - 0.5), 
\end{align}
}%
where $w\in [0,1]$ is the probability that two linked nodes have the same label and $\Gamma(u)$ is the set of neighbors of $u$. }

\subsection{Homophily-Entropy  Aggregator}
\label{sec:agg}
SybilBlind repeats  $k$ sampling trials, each of which produces a probability of being Sybil for every node. We denote the $k$ probabilities for $u$ as $ p_{1,u}, p_{2,u},\cdots, p_{k,u}$.  An aggregator is to reduce the $k$ probabilities to an \emph{aggregated} probability. 

\myparatight{Average, min, and max aggregators do not work well}  \emph{average}, \emph{min}, and \emph{max} aggregators are 
a few natural choices. Specifically, the average aggregator takes the average of the $k$ probabilities to be the aggregated one; the min aggregator is to take the minimum value of the $k$  probabilities, i.e., $p_u=\min_{i=1}^k p_{i,u}$; the max aggregator is to take the maximum value of the $k$ probabilities, i.e., $p_u=\max_{i=1}^k p_{i,u}$. However,  we demonstrated, theoretically and empirically, that these aggregators achieve performances that are the same with or even worse than random guessing. In particular, for the average aggregator, we can prove that the  expected aggregated probability is 0.5 for every node when the detector is SybilSCAR, 
which means that the expected performance of the average aggregator is the same as random guessing.
We show the proof in Appendix~\ref{averageaggregator}.


\myparatight{Our homophily-entropy aggregator (HEA)} We propose a novel aggregator based on two new metrics that we call \emph{homophily} and \emph{one-side entropy}. We observe that, when a sampling trial is a highly positively polarized scenario in which a majority of nodes in $B$ are benign and a majority of nodes in $S$ are Sybils, SybilSCAR can detect Sybils accurately. Our HEA aggregator aims to identify such sampling trials and use them to determine the aggregated probabilities. 
Next, we first formally define our homophily and one-side entropy metrics. 

Suppose in a sampling trial, SybilSCAR produces a probability of being Sybil for every node. We predict a node $u$ to be Sybil if $p_u > 0.5$, otherwise we predict  $u$ to be benign. Moreover, we denote by $s$ the fraction of nodes in the social network that are predicted to be Sybils. An edge $(u,v)$ in the social network is said to be \emph{homogeneous} if $u$ and $v$ have the same predicted label. Given these terms, we formally define homophily $h$ and one-side entropy $e$ as follows:
 \begin{align}
 &h=\frac{\#homogeneous\ edges}{\#edges\ in\ total} \nonumber \\
&e= 
\begin{cases}
0 &\text{ if } s > 0.5  \\
-s\text{log}(s) - (1-s)\text{log}(1-s) &\text{ otherwise } 
\end{cases}
\end{align}
Intuitively, homophily is the fraction of  edges that are predicted to be homogeneous. One-side entropy is small if too many or too few nodes  are predicted to be Sybils. 
In our threat model, 
we consider that the fraction of Sybils in the social network is less than 50\%. Therefore, we define one-side entropy to be 0 if more than a half of nodes are predicted to be Sybils. Note the difference between our defined one-side entropy and the conventional entropy in information theory.  

In a sampling trial that is an unpolarized scenario, we expect the homophily to be small because SybilSCAR tends to predict labels for nodes randomly. In a sampling trial that is a negatively polarized scenario, we expect the homophily to be large because a majority of  benign nodes are likely to be predicted to be Sybils and a majority of Sybils are likely to be predicted to be benign, which results in a large fraction of homogeneous edges. However, we expect the one-side entropy to be small because more than a half of nodes would be predicted to be Sybils. In a sampling trial that is a positively polarized scenario, we expect both homophily and one-side entropy to be large. 

Therefore, our  HEA aggregator aims to identify the sampling trials that have large homophily and one-side entropy. 
In particular, we first identify the top-$\kappa$ sampling trials among the $k$ sampling trials that have the largest homophily. Then, among the top-$\kappa$ sampling trials, we choose the sampling trial with the largest one-side entropy and use the probability obtained in this sampling trial as the aggregated probability. Essentially, among the top-$\kappa$ sampling trials, we identify the sampling trial with the largest $s$ that is no larger than 0.5, i.e., we aim to use the sampling trial that detects the most Sybils. Note that we can also reverse the order by first identifying the top-$\kappa$ sampling trials that have the largest one-side entropies and choose the sampling trial with the largest homophily. However, we find the performance is almost the same and we thus use the former way by default.

\section{Theoretical Analysis}
\label{sec:analysis}

\subsection{Sampling Size and Number of Sampling Trials} 
The sampler constructs a training set via assigning a label of benign to nodes in $B$ and a label of Sybil to nodes in $S$.  
We define label noise in the benign region (denoted as $\alpha_b$)  as the fraction of sampled nodes in the benign region whose assigned labels are Sybil. Similarly, we define label noise in the Sybil region (denoted as $\alpha_s$)  as the fraction of sampled nodes in the Sybil region whose assigned labels are benign. Formally, we have $\alpha_b=\frac{n_{sb}}{n_{sb} + n_{bb}}$ and $\alpha_s=\frac{n_{bs}}{n_{bs} + n_{ss}}$, where $n_{bb}$ and $n_{bs}$ respectively are the number of benign nodes and Sybils in $B$; $n_{sb}$ and $n_{ss}$ respectively are the number of benign nodes and Sybils in $S$.

We can derive an \emph{analytical form} for  the probability that label noise in both the benign region and the Sybil region are 
smaller than a threshold $\tau$ in a sampling trial.  Due to limited space, we omit the analytical form. 
However, the analytical form is too complex to illustrate the relationships between the sampling size and 
the number of sampling trials.  Therefore, we show the following theorem, which bounds the probability. 


 
\setcounter{theorem}{0}
\begin{theorem}
\label{theorem:bound}
In a sampling trial with a sampling size of $n$,  the probability that label noise in both the benign region and the Sybil region are no bigger than $\tau$ ($\tau \leq 0.5$) is bounded as
{\footnotesize
\begin{align}
\label{bound}
(1-r)^n r^n \leq \text{Pr}(\alpha_b \leq \tau, \alpha_s \leq \tau) \leq \exp \big(-\frac{2(1-2\tau)^2(1-r)^2n}{\tau^2 + (1-\tau)^2}\big),
\end{align} 
}%
where $r$ is the fraction of Sybils in the social network.
\end{theorem}
\begin{proof}
See Appendix~\ref{proofoftheorembound}.
\end{proof}

\myparatight{Implications of Theorem~\ref{theorem:bound}}
Suppose in a social network, SybilSCAR is robust to label noise upto $\tau$, i.e., its performance almost does not degrade when the noise level is $\tau$, then SybilBlind requires at least one sampling trial, in which the label noise is less than or equal to $\tau$, to detect Sybils accurately. We have several qualitative implications from Theorem~\ref{theorem:bound}.  
We note that these implications also hold when using the analytical form of the probability that label noise are 
smaller than $\tau$. Here, we choose Theorem~\ref{theorem:bound} 
because of its conciseness. 

First,  when the sampling size is $n$ and SybilSCAR is robust to label noise up to $\tau$ in the social network,
the expected number of sampling trials (i.e., $k$) that SybilBlind requires is bounded as $k_{min} \leq k \leq k_{max}$, where 
$k_{min}=\text{exp}\big(\frac{2(1-2\tau)^2(1-r)^2n}{\tau^2 + (1-\tau)^2}\big)$ 
and $k_{max} = \frac{1}{(1-r)^n r^n}$.
Note that $k_{min}$ is exponential with respect to $n$, which could be very large even if $n$ is moderate. \emph{However, through empirical evaluations, we found $k$ can be largely reduced and a moderate $k$ could make SybilBlind obtain satisfying performance.}  
Second, when $\tau$ gets bigger, $k_{min}$ gets smaller, which implies that SybilBlind tends to require less sampling trials when detecting Sybils in a social network in which SybilSCAR can tolerate larger label noise.
Third, we observe a \emph{scale-free} property, i.e., 
the number of sampling trials is not related to the size (i.e., $|V|$ or $|E|$) of the social network.


\alan{
\subsection{Complexity Analysis}
\label{sec:complexity}

\myparatight{Space and time complexity} The major space cost of SybilBlind consists of storing the social network and storing the top-$\kappa$ vectors of posterior probabilities.
SybilBlind uses an adjacency list to represent the social network, with the space complexity $O(2|E|)$, and stores the top-$\kappa$ vectors of posterior probabilities of all nodes. Therefore, the space complexity of SybilBlind is $O(2|E|+\kappa |V|)$. 

In each trial and in each iteration, SybilBlind applies a local rule to every node, and the time complexity of the local rule 
to a node $u$ with $|\Gamma_u|$ friends is $O(|\Gamma_u|)$. Therefore, the time complexity of SybilBlind in one iteration is $O(|E|)$.
Since SybilBlind performs $k$ sampling trials and each trial runs $T$ iterations, it thus has a time complexity of $O(kT|E|)$.

\myparatight{Two-level parallel implementation} We can have a two-level parallel implementation of SybilBlind on a data center which is  a standard backend for  social web services. First, different sampling trials can be run on different machines. They only need to communicate once to share their vectors of posterior probabilities.  Second, each machine can parallelize SybilSCAR using multithreading. Specifically, in each iteration of SybilSCAR, each thread  applies the local rule to a subset of nodes in the social network. 
}


\section{Experiments}

\subsection{Experimental Setup}
\label{exp-setup}
\myparatight{Datasets}
We use social networks with synthesized Sybils and Twitter datasets with real Sybils for evaluations. 
Table~\ref{dataset} summarizes the datasets. 

{{\bf 1) Social networks with synthesized Sybils.}}
Following previous works
~\cite{Yu08,Danezis09,sybilrank}, we use a real-world social network as the benign region, while synthesizing the Sybil region and attack edges. 
Specifically, we take a Facebook network as the benign region; we synthesize the Sybil region using the \emph{Preferential Attachment (PA)} model~\cite{Barabasi99}, which is a widely used method to generate networks; and we add attack edges between the benign region and the Sybil region uniformly at random. 
In this graph, nodes 
are Facebook users and  two nodes are connected if they
are friends.
We synthesize the Sybil region such that 20\% of users in the social network are Sybils; the average degree in the Sybil region is the same as that in the benign region in order to avoid asymmetry between the two regions introduced by density. We set the number of attack edges as 500, and thus the average attack edge per Sybil is 0.06.

{{\bf 2) Small Twitter with real Sybils.}}
We obtained a publicly available Twitter dataset with 809 Sybils and 7,358 benign nodes from Yang et al.~\cite{Yang12-spam}.  A node is a Twitter user and an edge means two users follow each other. Sybils were labeled spammers. 9.9\% of nodes are Sybils and 53.4\% of Sybils are connected. The average degree is 16.72, and the average attack edge per Sybil is 49.46. 


\begin{table}[t]\renewcommand{\arraystretch}{1.2}
\centering
\caption{Dataset statistics.}
\centering
\begin{tabular}{|c|c|c|c|} \hline 
{\small Metric} & {\small Facebook} & {\small Small Twitter} &{\small Large Twitter}\\ \hline
{\small \#Nodes} & {\small 43,953} & {\small 8,167} & {\small 41,652,230}\\ \hline
{\small \#Edges} & {\small 182,384} & {\small 68,282} & {\small 1,202,513,046}\\ \hline
{\small Ave. degree} &  {\small 8.29} & {\small 16.72} & {\small 57.74}\\ \hline
{\small Ave. \#attack edge per Sybil} &  {\small 0.06} & {\small 49.46} & {\small 181.55}\\ \hline
\end{tabular}
\vspace{-4mm}
\label{dataset}
\end{table}

{{\bf 3) Large Twitter with real Sybils.}} 
We obtained a snapshot of a large-scale Twitter follower-followee network crawled by Kwak et al.~\cite{kwak2010twitter}. A node is a Twitter user and an edge between two nodes means that one node follows the other node. The network has 41,652,230 nodes and 1,202,513,046 edges. 
To perform evaluation, we need ground truth labels of the nodes. Since the Twitter network includes users' Twitter IDs, we wrote a crawler to visit each user's profile using Twitter's API, which tells us the status (i.e., active, suspended, or deleted)  of each user. In our ground truth, 205,355 nodes
were suspended, 5,289,966 nodes 
were deleted, and the remaining 36,156,909 nodes are active. 
We take suspended users as Sybils 
and active users as benign nodes. 
85.3\% Sybils are connected with an average degree 24.
1.5\% of the total edges are attack edges and the average number of attack edges per Sybil is 181.55. 
We acknowledge that our ground truth labels might be noisy since some active users might be Sybils, but they evaded Twitter's detection, and Twitter might have deleted some Sybils. 
\myparatight{AUC as an evaluation metric}   Similar to previous studies~\cite{sybilrank,integro,sybilbelief,sybilscar}, we use the Area Under the Receiver Operating Characteristic Curve (AUC) as an evaluation metric.
Suppose we rank nodes according to their probabilities of being Sybil in a descending order. AUC is the probability that a randomly selected Sybil ranks higher than a randomly selected benign node. Random guessing, which ranks nodes uniformly at random, achieves an AUC of 0.5.

\myparatight{Compared methods} We adapt a community detection method and SybilSCAR to detect Sybils when no manual labels are available.
 \Alan{Moreover, we compare with SybilRank~\cite{sybilrank} and SybilBelief~\cite{sybilbelief} that require manual labels.}

{{\bf 1) Community detection (Louvain Method).}} When there are no manually labeled training sets, community detection seems to be a natural choice to detect connected Sybils.\footnote{The local community detection method~\cite{Viswanath10}
requires labeled benign nodes and thus is inapplicable to detect Sybils without a manually labeled training set.}  A community detection method divides a social network into connected components (called ``communities"), where nodes in the same community are densely connected while nodes across different communities are loosely connected. Presumably, Sybils are in the same communities. 

 Since the benign region itself often consists of multiple communities~\cite{sybilrank,alvisiSybil13}, the key challenge of  community detection methods  is how to determine which communities correspond to Sybils. 
  Assigning a label of Sybil (or benign) to a community means that all nodes in the community are Sybils (or benign).
Since it is unclear how to assign labels to the communities algorithmically (though one could try various heuristics), in our experiments, 
we {assume} one could label communities such that  community detection achieves a \emph{false negative rate} that is the closest to that of SybilBlind. 
Specifically, SybilBlind predicts a node to be Sybil if its aggregated probability is larger than 0.5, and thus we can compute the false negative rate for SybilBlind. 
Then we compare  community detection with SybilBlind with respect to AUC, via ranking the communities labeled as Sybil higher than those labeled as benign.
 Our experiments give advantages to community detection since this label assignment might not be found in practice.  
Louvain method~\cite{Louvain08} is a widely used community detection method, 
which is efficient and outperforms a variety of  community detection methods~\cite{Louvain08}. 
Therefore, we choose Louvain method in our experiments.

{{\bf 2) SybilSCAR with a sampled noisy training set (SybilSCAR-Adapt).}} 
When a manually labeled training set is unavailable, we use our sampler to sample a training set and treat it as the input to SybilSCAR.
 The performance of this adapted SybilSCAR highly depends on the label noise of the training set. 

\Alan{{\bf 3) SybilRank and SybilBelief with labeled training set.} 
SybilRank~\cite{sybilrank} and SybilBelief~\cite{sybilbelief} are state-of-the-art random walk-based method and LBP-based method, respectively. SybilRank can only leverage labeled benign nodes, while SybilBelief can leverage both labeled benign nodes and labeled Sybils.  
We randomly sample a labeled training set, where the number of labeled benign nodes and Sybils equals $n$ (the sampling size of SybilBlind).}


{{\bf 4) SybilBlind.}}
 In the Facebook network with synthesized Sybils, our sampler samples the two subsets $B$ and $S$ uniformly at random from the entire social network.  
For the Twitter datasets, directly sampling two subsets $B$ and $S$ with a low label noise is challenging due to the number of benign nodes is far larger than that of Sybils. Thus, we refine our sampler by using discriminative node features. Previous studies~\cite{Yang11-sybil,Yang12-spam} found that Sybils proactively follow a large number of benign users in order to make more benign users follow them, but only a small fraction of benign users will follow back.
Therefore, we extract the \emph{follow back rate (FBR)} feature for each node in the Twitter datasets.
Then we rank all nodes according to their FBR features in an ascending order. 
Presumably, some Sybils are ranked high and some benign nodes are ranked low in the ranking list. 
Thus, we sample the subset $B$ from the bottom-$K$ nodes and sample the subset $S$ from the top-$K$ nodes. 
Consider the different sizes of the two Twitter datasets, we set $K=1,000$ and $K=500,000$ in the small and large Twitter datasets, respectively. 
This sampler is more likely to sample training sets that have lower label noise, and thus it improves SybilBlind's performance. \emph{Note that when evaluating SybilSCAR-Adapt on the Twitter datasets, we also use FBR-feature-refined sampler to sample a training set.}
As a comparison, we also evaluate the method simply using the FBR feature and denote it as FBR.
\Alan{Moreover, we evaluate SybilBlind with randomly sampled two subsets without the FBR feature, which we denote as SybilBlind-Random.} 

%


\myparatight{Parameter settings} 
For SybilBlind, according to Theorem 1, the minimal number of sampling trials $k_{min}$ to generate a training set with label noise less than or equal to $\tau$ is exponential with respect to $n$, and $k_{min}$ would be very large even with a modest $n$.  
However, through empirical evaluations, we found that the number of sampling trials can be largely decreased when using the FBR-feature-refined sampler. 
Therefore, we instead use the following heuristics to set the parameters, with which SybilBlind has already obtained satisfying performance.  
Specifically,
$n=10$, $k=100$, and $\kappa=10$ for the Facebook network with synthesized Sybils; $n=100$, $k=20$, and $\kappa=10$ for the small Twitter; and $n=100,000$, $k=20$, and $\kappa=10$ for the large Twitter. We use a smaller $k$ for Twitter datasets because FBR-feature-refined sampler is more likely to sample training sets with smaller label noise. We use a larger sampling size $n$ for the large Twitter dataset because its size is much bigger than the other two datasets. 
\Alan{We will also explore the impact of parameters and the results are shown in Figure~\ref{parameter}.}


For other compared methods, we set parameters according to their authors. 
For instance, we set $\theta=0.4$ for SybilSCAR. 
SybilRank requires early termination, and its number of iterations is suggested to be $O(\log |V|)$.
For each experiment, we repeat 10 times and compute the average AUC. 
We implement SybilBlind in C++ using multithreading, and we obtain the publicly available implementations for SybilSCAR (also in C++)\footnote{\url{http://home.engineering.iastate.edu/~neilgong/dataset.html}} and Louvain method\footnote{\url{https://sites.google.com/site/findcommunities/}}. 
We perform all our experiments on a Linux machine with 512GB memory and 32 cores. 


%

\begin{figure}[!tbp]
  \centering
  \begin{minipage}[c]{0.32\textwidth}
    \includegraphics[width=\textwidth]{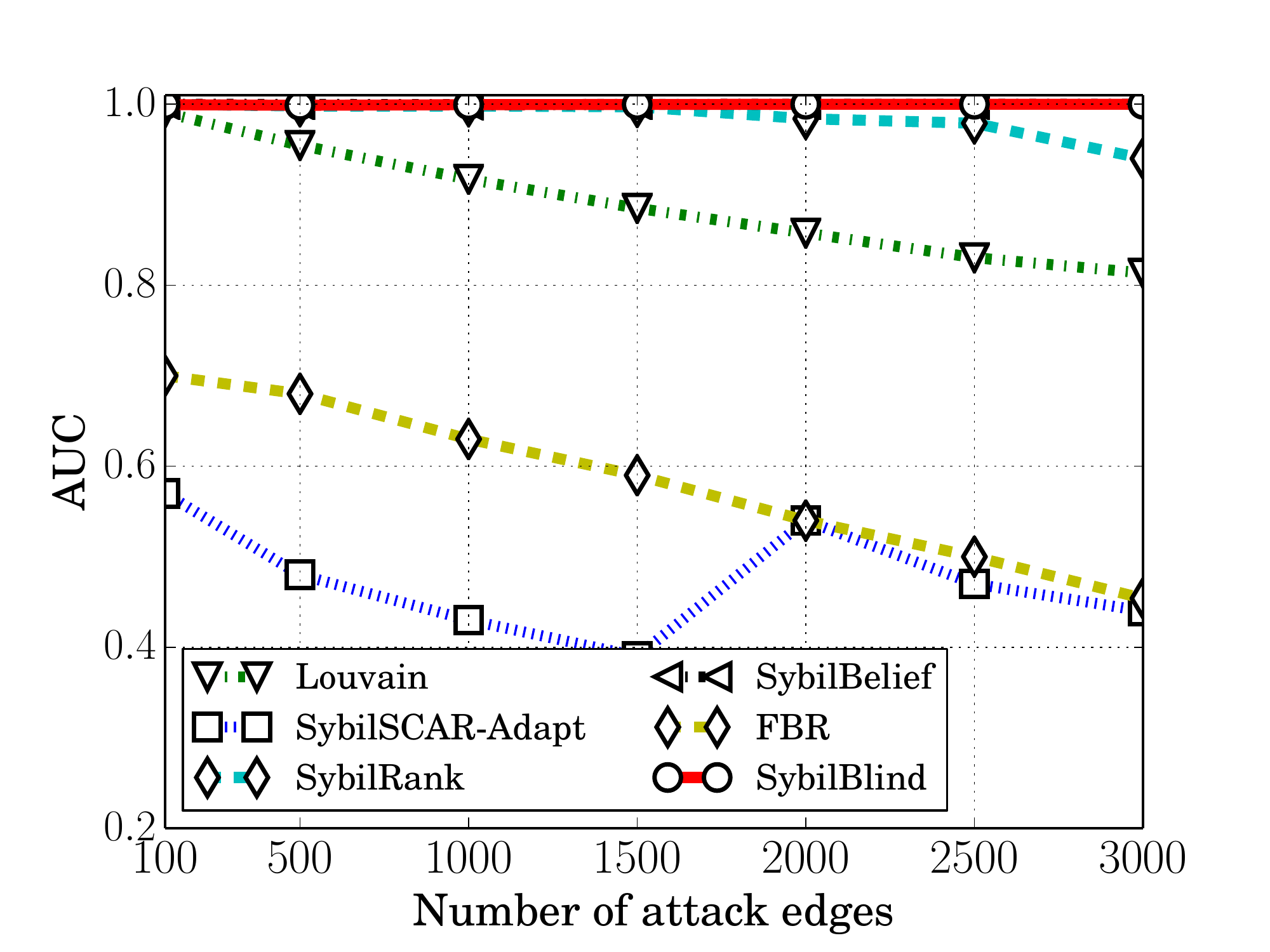}
    \caption{AUCs on the Facebook network with synthesized Sybils. SybilBlind is robust to various numbers of attack edges.}
    \label{louvain_ACU}
  \end{minipage}
  \quad
  \begin{minipage}[c]{0.64\textwidth}
  	\center
	\vspace{-4mm}
	\subfloat[Sampling size]{\includegraphics[width=0.48\textwidth]{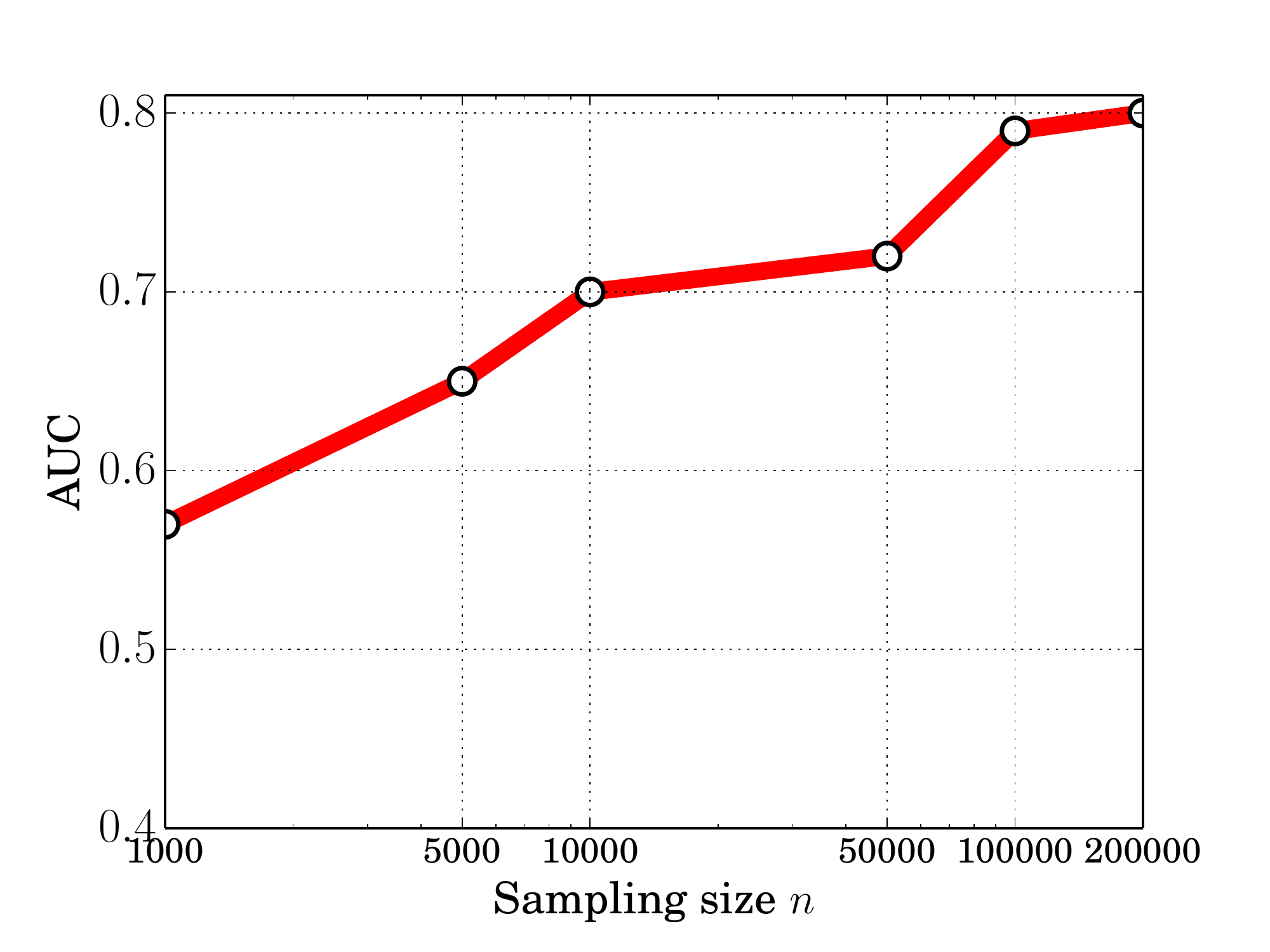}\label{effect_size}}
	\subfloat[Sampling trial]{\includegraphics[width=0.48\textwidth]{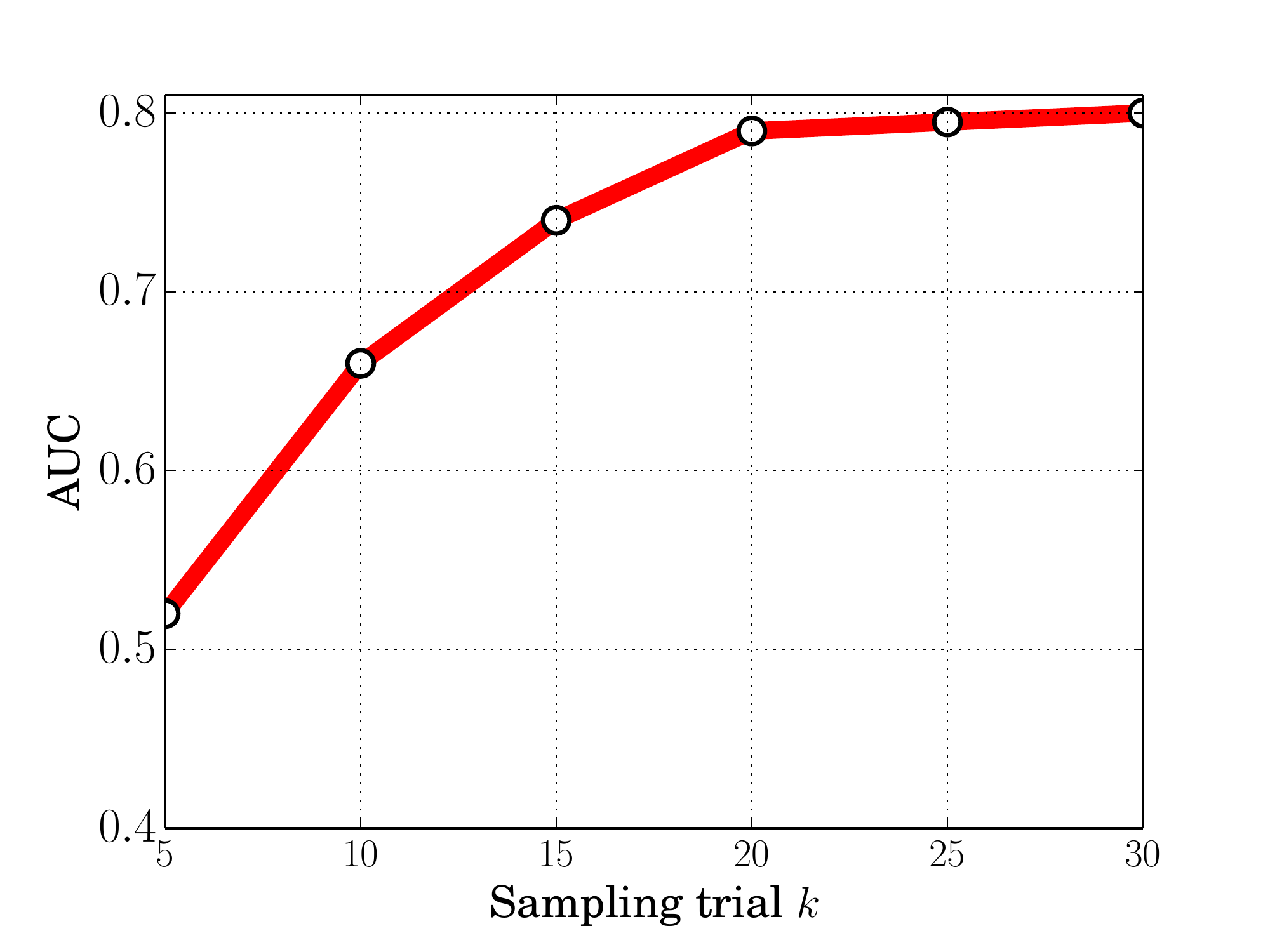} \label{effect_trial}}
	\caption{AUCs of SybilBlind vs. (a) sampling size $n$ and (b) number of sampling trials $k$ on the large Twitter. We observe that SybilBlind achieves high AUSs when $n$ and $k$ reach certain values.}
	\label{parameter}
	\end{minipage}
  \vspace{-4mm}
\end{figure}


\begin{table}[!tbp]\renewcommand{\arraystretch}{1.2}
\centering
\caption{AUCs of the compared methods on the Twitter datasets.}
\centering
\begin{tabular}{|c|c|c|} \hline 
{\small \textbf{Method}} & {\small \textbf{Small Twitter}} & {\small \textbf{Large Twitter}}  \\ \hline
{\small Louvain} & 	{\small 0.54} & {\small 0.50} \\ \hline
{\small SybilSCAR-Adapt} & {\small 0.89} & {\small 0.70} \\ \hline
{\small SybilRank} & {\small 0.86} & {\small 0.69} \\ \hline
{\small SybilBelief} & {\small 0.98} & {\small 0.78} \\ \hline
{\small FBR} & {\small 0.60} & {\small 0.51} \\ \hline
{\small SybilBlind-Random} & {\small 0.82} & {\small 0.65} \\ \hline
{\small SybilBlind} & {\small 0.98} & {\small 0.79} \\ \hline
\end{tabular}
\vspace{-4mm}
\label{AUC_Twitter}
\end{table}

\subsection{Results}
\label{expSybilBlind}

\begin{figure*}[t]
\center
\vspace{-4mm}
\subfloat[Small Twitter]{\includegraphics[width=0.42\textwidth]{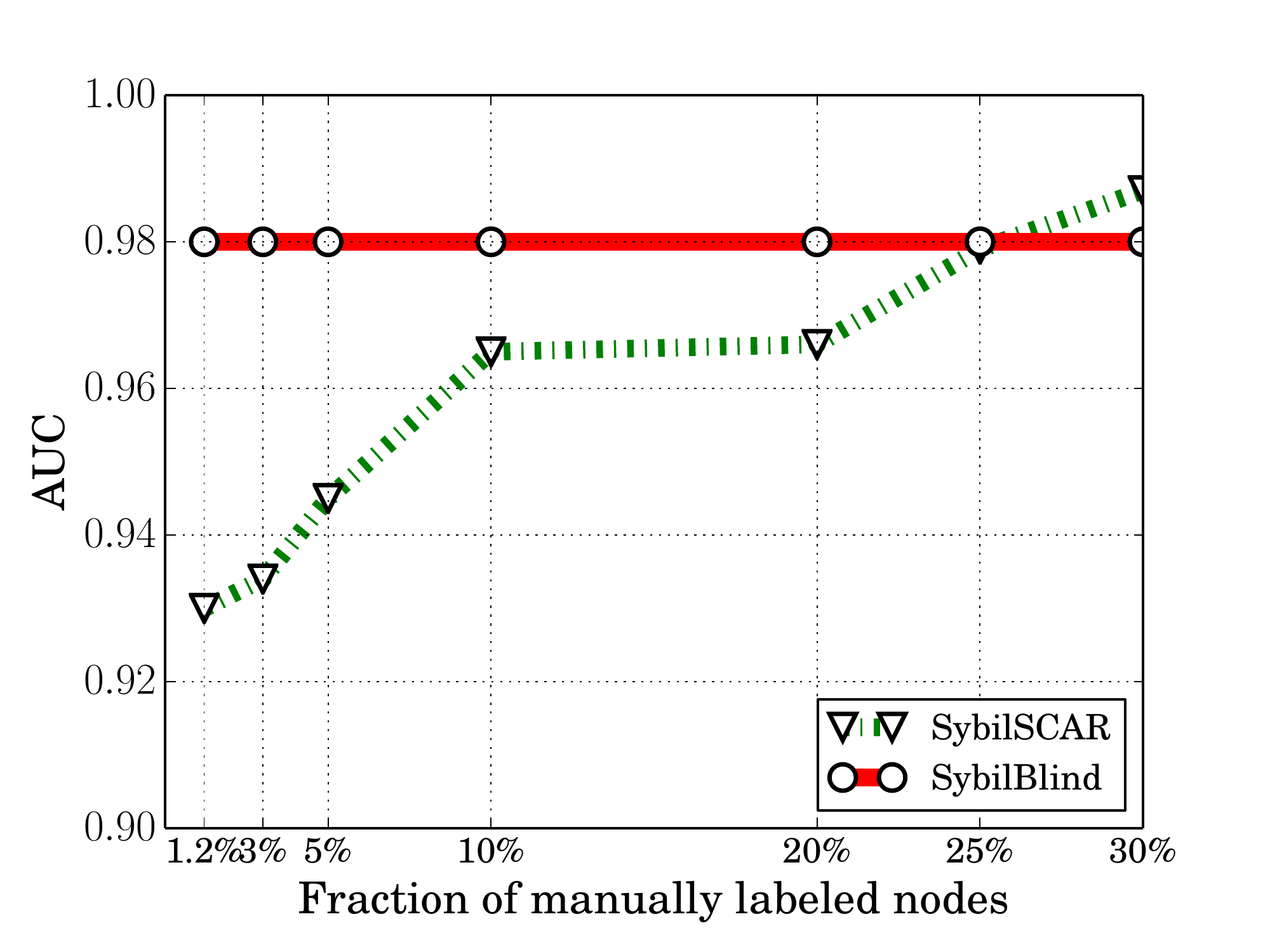}}
\subfloat[Large Twitter]{\includegraphics[width=0.42\textwidth]{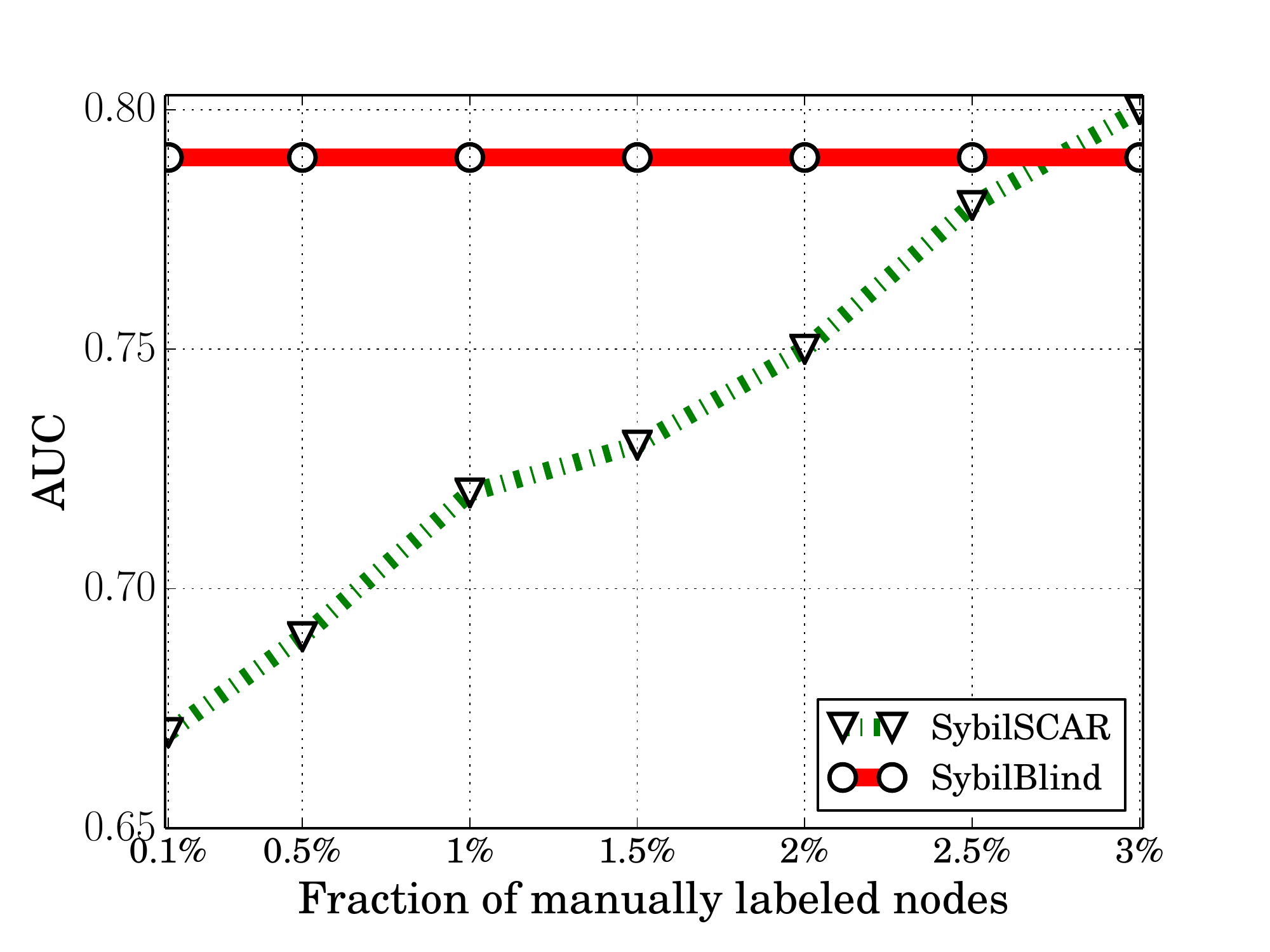}}
\caption{AUCs of SybilSCAR vs. the fraction of nodes that are manually labeled as a training set on the small Twitter and large Twitter datasets. We observe that SybilSCAR requires manually labeling about 25\% and 2.8\% of total nodes on the small Twitter and large Twitter datasets to be comparable to SybilBlind.}
\label{stability}
\vspace{-6mm}
\end{figure*}

\myparatight{AUCs of the compared methods} Figure~\ref{louvain_ACU} shows AUCs of the compared methods on the Facebook network with synthesized Sybils as we increase the number of attack edges. 
\Alan{Note that SybilBlind-Random is essentially SybilBlind in this case, as we randomly sample the subsets without the FBR feature.}
 Table~\ref{AUC_Twitter} shows AUCs of the compared methods for the Twitter datasets with real Sybils. 
We observe that 
1) SybilBlind outperforms Louvain method. Specifically, when the number of attack edges gets relatively large, even if one could design an algorithm to label communities such that Louvain method can detect as many Sybils as SybilBlind (i.e., similar false negative rates), Louvain method will rank a large fraction of benign users higher than Sybils, resulting in small AUCs.
The reason is that some communities include a large number of both benign users and Sybils, which is an intrinsic limitation of community detection.  
2) SybilBlind outperforms SybilSCAR-Adapt, which validates that our homophily-entropy aggregator is significant and essential. Thus, aggregating results in multiple sampling trials can boost the performance.  
\Alan{3) SybilBlind outperforms SybilRank and is comparable with SybilBelief, even if SybilRank and SybilBelief use a labeled training dataset. 
This is because the FBR-feature-refined sampler can sample training sets with relatively small label noise and SybilSCAR is robust to such label noise. As SybilSCAR was shown to outperform SybilRank and be comparable with SybilBelief~\cite{sybilscar}, so does SybilBlind.} 
4) SybilSCAR-Adapt achieves AUCs that are close to random guessing on the Facebook network. This is because the sampled training set has random label noise that could be large. 
SybilSCAR-Adapt works better on the Twitter datasets. Again, this is because the FBR feature assists our sampler to obtain the training sets with small label noise on the Twitter datasets and SybilSCAR can tolerate such label noise.  
5) FBR achieves a small AUC. This indicates that although the FBR feature can be used to generate a ranking list with small label noise by treating top-ranked nodes as Sybils and bottom-ranked nodes as benign, the overall ranking performance on the entire nodes is not promising. 
\Alan{6) SybilBlind-Random's performance decreases on the Twitter datasets. The reason is that it is difficult to sample training sets with small label noise, as the number of benign nodes is far larger than the number of Sybils on the Twitter datasets.}    

\myparatight{Number of manual labels SybilSCAR requires to match SybilBlind's performance} Intuitively, given a large enough manually labeled training set, SybilSCAR that takes the manually labeled training set as an input would outperform SybilBlind. Therefore, one natural question is how many nodes need to be manually labeled in order for SybilSCAR to match SybilBlind's performance. To answer this question, we respectively sample $x$ fraction of total nodes in the small Twitter dataset and large Twitter dataset and treat them as a manually labeled training set, i.e., the benign nodes are assigned a label of benign and the Sybils are assigned a label of Sybil. Note that the manually labeled training set has no label noise. Then, we run SybilSCAR with the training set, rank the remaining nodes using their probabilities of being Sybil, and compute an AUC.  
Figure~\ref{stability} shows the AUCs of SybilSCAR  as we increase $x$ from 0.1\% to 3\% on the small Twitter and large Twitter datasets. For comparison, we also show the AUC of SybilBlind on the small Twitter and large Twitter datasets, which is a straight line since it does not rely on the manually labeled training set. 
We observe that SybilSCAR requires manually labeling about 25\% of total nodes on the small Twitter and about 2.8\% of total nodes on the large Twitter in order to achieve an AUC that is comparable to SybilBlind. 

\begin{figure}[!tbp]
  \centering
  \begin{minipage}[b]{0.45\textwidth}
    \includegraphics[width=\textwidth]{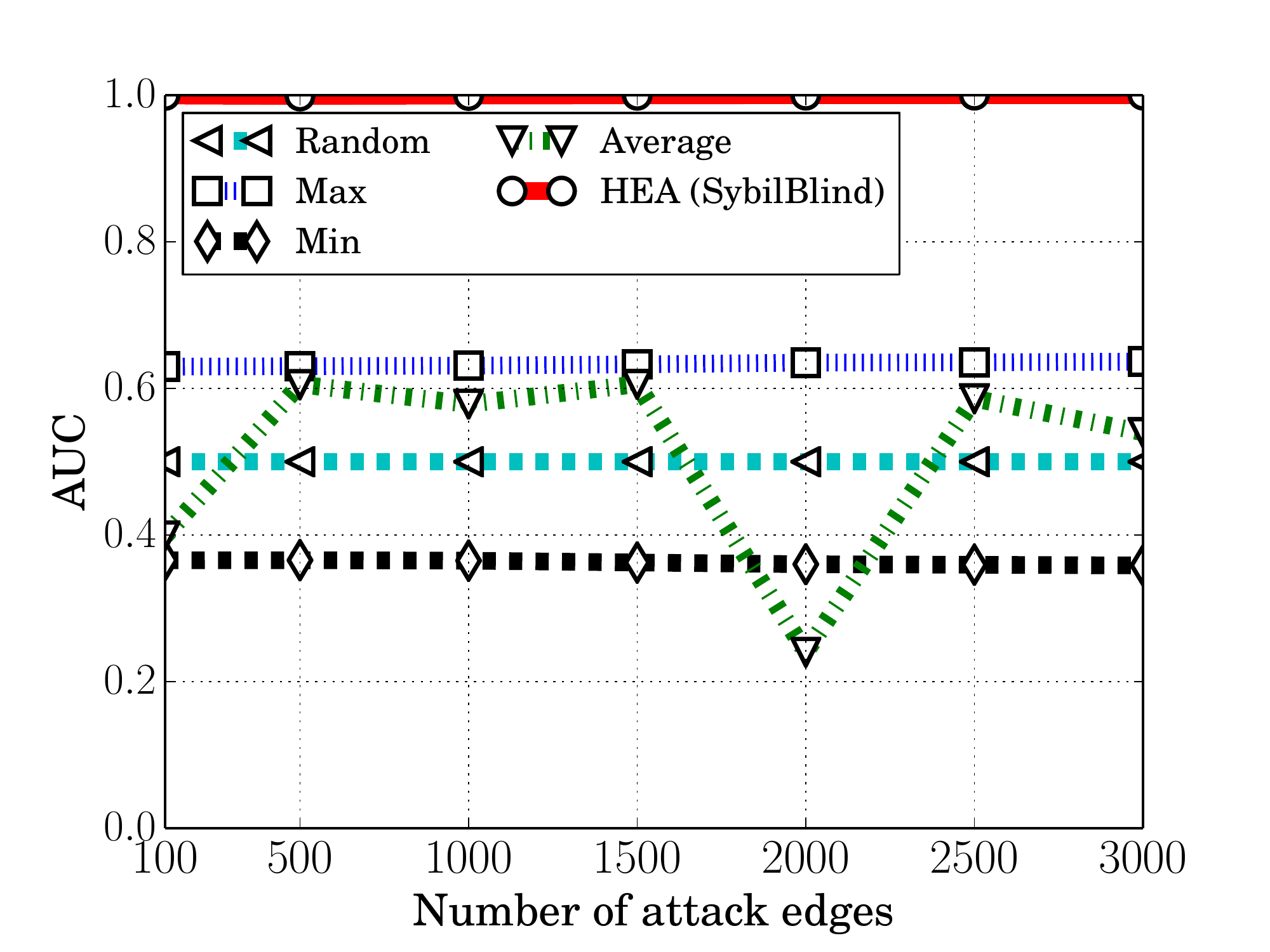}
    \caption{Performance of different aggregators on the Facebook network with synthesized Sybils. Our homophily-entropy aggregator (HEA) significantly outperforms the average, min, and max aggregators.}
    \label{aggregator_AUC}
  \end{minipage}
  \hfill
  \begin{minipage}[b]{0.45\textwidth}
    \includegraphics[width=\textwidth]{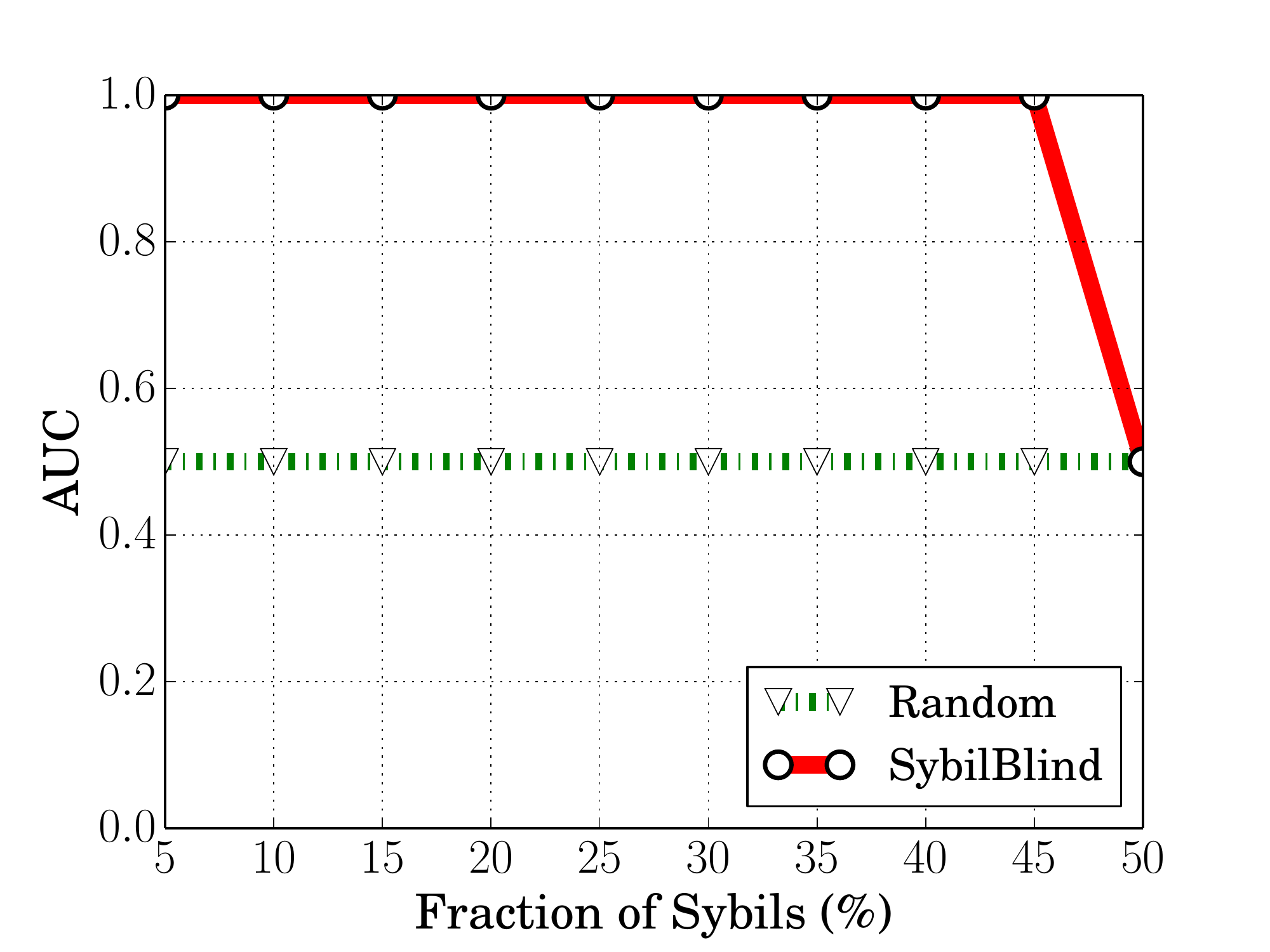}
    \caption{Impact of the fraction of Sybils on the Facebook network. We observe that SybilBlind can accurately detect Sybils once the fraction of Sybils is smaller than 50\%, i.e., Sybils are less than benign nodes.}
    \label{attackedge_AUC}
  \end{minipage}
  \vspace{-4mm}
\end{figure}

\myparatight{Comparing different aggregators} 
 Figure~\ref{aggregator_AUC} shows the performances of different aggregators on the Facebook network with synthesized Sybils as we increase the number of attack edges. We observe that our homophily-entropy aggregator (HEA) significantly outperforms the average, min, and max aggregators.  The average aggregator achieves performances that are close to random guessing. 
 This is because the average aggregator assigns an expected aggregated probability of 0.5 to every node. Moreover, the min aggregator achieves AUCs that are worse than random guessing, while the max aggregator achieves AUCs that are slightly higher than random guessing. 
 It is an interesting future work to theoretically understand the performance gaps for the min and max aggregators. 

\myparatight{Impact of the fraction of Sybils}  
Figure~\ref{attackedge_AUC} shows the AUCs of SybilBlind as the social network has more and more Sybils. We performed the experiments on the Facebook  network with synthesized Sybils since we need social networks with different number of Sybils. The number of attack edges is set to be 500. We observe that SybilBlind can accurately detect Sybils (AUCs are close to 1) once the fraction of Sybils is smaller than 50\%, i.e., Sybils are less than benign nodes. 
We note that when Sybils are more than benign nodes, SybilBlind would rank benign nodes higher than Sybils, resulting in AUCs that are close to 0. However, in practice, Sybils are less than benign nodes, as we discussed in Section~\ref{threatmodel}.

\Alan{
\myparatight{Impact of $n$ and $k$} 
Figure~\ref{effect_size} and~\ref{effect_trial} show AUCs of SybilBlind vs. sampling size $n$ ($k=20$) and the number of sampling trials $k$ ($n=100,000$) on the large Twitter, respectively. 
We observe that the AUCs increase as the sampling size and the number of sampling trials increase. The AUCs become stable after $n$ and $k$ reach certain values. The AUCs are small when $n$ or $k$ is small, because it is harder to sample  training sets with relatively small label noise.}

\Alan{
\myparatight{Running time} 
We show running time of SybilBlind on the large Twitter. 
We concurrently generate sampling trials using multiprocessing. In particular, we create 4 processes in parallel, each of which runs one sampling trial. Moreover, each sampling trial runs SybilSCAR using multithreading (20 threads in our experiments). It took about 2 hours for one process to run SybilSCAR in one sampling trial, and the total time for our SybilBlind with 20 sampling trials is around 10 hours. 
}


\section{Conclusion and Future Work}
We design a novel structure-based framework called SybilBlind to detect Sybils
in online social networks without a manually labeled training dataset. 
We demonstrate the effectiveness of SybilBlind using 
both social networks with synthetic Sybils and Twitter datasets with real Sybils. 
Our results show that Sybils can be detected without manual labels. 
Future work includes applying SybilBlind to detect Sybils 
{with sampled subsets with different sizes} 
and extending SybilBlind to learn general machine learning classifiers without manual labels.   

\Alan{
\myparatight{Acknowledgements} We thank the anonymous reviewers and our shepherd Jason Polakis for their constructive comments. This work was supported by NSF under grant CNS-1750198 and a research gift from JD.com.
}

\bibliographystyle{splncs04}
\bibliography{refs}

\appendix

\section{Performance of the Average Aggregator}
\label{averageaggregator}
\begin{theorem}
When SybilBlind uses the average aggregator, the expected aggregated probability is 0.5 for every node. 
\end{theorem}

\begin{proof}
Suppose in some sampling trial, the sampled subsets are $B$ and $S$,  and SybilSCAR halts after $T$ iterations. 
We denote by $q_u$ the prior probability and by ${p_u}^{(t)}$ the probability in the $t$th iteration for $u$, respectively. Note that the subsets $B'=S$ and $S'=B$ are sampled by the sampler with the same probability. We denote by $q_u'$ the prior probability and by ${p_u}^{(t)'}$ the probability in the $t$th iteration for $u$, respectively, when SybilSCAR uses the subsets $B'$ and $S'$. We prove that ${q}_u'=1 -{q}_u$ and ${p}_u^{(t)'} = 1-{p}_u^{(t)}$ for every node $u$ and iteration $t$. 
First, we have:
{\small
\begin{align}
{q}_u' = 
\begin{cases} 
0.5 - \theta &= 1-{q}_u \text{\ \ if } u \in S \\ \nonumber
0.5 + \theta &= 1-{q}_u \text{\ \ if } u \in B \\ \nonumber
0.5 &= 1-{q}_u \text{\ \ otherwise,} \nonumber
\end{cases} 
\end{align}
}%
which means that ${{q}_u}'=1-{q}_u$ for every node.

We have ${p_u}^{(0)'} = {q_u}'$ and ${p_u}^{(0)} = {q_u}$. Therefore, ${p}_u^{(0)'} =1 -{p}_u^{(0)}$ holds for every node in the $0$th iteration. We can also show that  ${p}_u^{(t)'} = 1-{p}_u^{(t)}$ holds for every node in the $t$th iteration if ${p}_u^{(t-1)'} = 1-{p}_u^{(t-1)}$ holds for every node. Therefore, ${p}_u^{(t)'} = 1-{p}_u^{(t)}$ holds for every node $u$ and iteration $t$. 
%
As a result, with the sampled subsets $B'$ and $S'$, SybilSCAR also halts after $T$ iterations.  Moreover, the average probability  in the two sampling trials (i.e., the sampled subsets are $B$ and $S$, and $B'=S$ and $S'=B$) is 0.5 for every node. For each pair of sampled subsets $B$ and $S$, there is a pair of subsets $B'=S$ and $S'=B$ that  are sampled by our sampler with the same probability. Therefore, the expected aggregated probability is 0.5 for every node.
\end{proof}

\section{Proof of Theorem 1}
\label{proofoftheorembound}

\myparatight{Lower bound}
We have:
{\small
\begin{align}
\text{Pr}(\alpha_b \leq \tau, \alpha_s \leq \tau) &\geq \text{Pr}(\alpha_b=\alpha_s = 0) =(1-r)^n r^n.
\end{align}
}%

We note that this lower bound is very loose because we simply ignore the cases where $\text{Pr}(0<\alpha_b \leq \tau, 0<\alpha_s \leq \tau)$. 
However, this lower bound is sufficient to give us qualitative understanding.

\myparatight{Upper bound}
 We observe that the probability that  label noise in both the benign region and the Sybil region are no bigger than $\tau$ is bounded by the probability that label noise in the benign region or the Sybil region is no bigger than $\tau$. Formally, we have:
{\small
\begin{align}
\label{min}
\text{Pr}(\alpha_b \leq \tau, \alpha_s \leq \tau) \leq \min \{\text{Pr}(\alpha_b \leq \tau), \text{Pr}( \alpha_s \leq \tau) \} 
\end{align}
}%
Next, we will bound the probabilities $\text{Pr}(\alpha_b \leq \tau)$ and $\text{Pr}( \alpha_s \leq \tau)$ separately.  We will take $\text{Pr}(\alpha_b \leq \tau)$ as an example to show the derivations, and similar derivations can be used to bound  $\text{Pr}( \alpha_s \leq \tau)$. 

We observe the following equivalent equations:
{\small
\begin{align}
\label{equi}
\text{Pr}(\alpha_b \leq \tau)&=\text{Pr}(\frac{n_{sb}}{n_{sb} + n_{bb} } \leq \tau) =\text{Pr}(\tau n_{bb} + (\tau - 1) n_{sb} \geq 0) 
 \end{align}
}%
We define $n$ random variables $X_1, X_2, \cdots, X_n$  and $n$ random variables $Y_1, Y_2, \cdots, Y_n$ as follows:
{\small
\begin{align}
&X_i = 
\begin{cases}
\tau &\text{\ \ \ \ \ if the } i \text{th node in } $B$ \text{ is benign} \nonumber \\
0 &\text{\ \ \ \ \ otherwise}  \\
\end{cases} \\
&Y_i = 
\begin{cases}
\tau-1 &\text{ if the } i \text{th node in } $S$ \text{ is benign} \nonumber \\
0 &\text{ otherwise,}  
\end{cases}
 \end{align}
 }%
 where $i=1,2,\cdots, n$. According to our definitions, we have 
 $\text{Pr}(X_i=\tau)=1 - r$ and  $\text{Pr}(Y_i=\tau - 1)=1 - r$,  where $i=1,2,\cdots, n$.
 Moreover, we denote $S$ as the sum of these random variables, i.e., 
 $S=\sum_{i=1}^n X_i + \sum_{i=1}^n Y_i$. Then, the expected value of $S$ is
 $E(S)=-(1-2\tau)(1-r)n$.  With the variables $S$ and $E(S)$, we can further rewrite 
 Equation~\ref{equi} as follows:
 {\small
 \begin{align}
\text{Pr}(\alpha_b \leq \tau)= \text{Pr}( S -E(S) \geq -E(S)) \nonumber
 \end{align}
 }%
 According to Hoeffding's inequality~\cite{Hoeffding63}, we have
 {\footnotesize
   \begin{align}
\text{Pr}( S -E(S) \geq -E(S)) &\leq \text{exp}\Big(-\frac{2E^2(s)}{(\tau^2 + (1-\tau)^2)n}\Big) =\text{exp}\Big(-\frac{2(1-2\tau)^2(1-r)^2n}{\tau^2 + (1-\tau)^2}\Big) \nonumber
 \end{align}
 }%
Similarly, we can derive an upper bound of $Pr( \alpha_s \leq \tau)$ as follows:
{\small
\begin{align}
\text{Pr}( \alpha_s \leq \tau) &\leq \text{exp}\Big(-\frac{2(1-2\tau)^2 r^2 n}{\tau^2 + (1-\tau)^2}\Big) 
 \end{align}
}%
Since we consider $r<0.5$ in this work, we have:
{\small
\begin{align}
\label{minvalue}
 \min \{\text{Pr}(\alpha_b \leq \tau), \text{Pr}( \alpha_s \leq \tau) \} = \text{exp}\Big(-\frac{2(1-2\tau)^2(1-r)^2n}{\tau^2 + (1-\tau)^2}\Big) 
\end{align}
}

By combining Equation~\ref{min} and Equation~\ref{minvalue}, we obtain Equation~\ref{bound}.

\end{document}